\documentclass[11pt]{article}
\usepackage{fullpage, thm-restate,thmtools}

\usepackage{times}
\usepackage{comment,amsfonts,amssymb,amsmath,amsthm,graphicx,algorithm,algorithmic}
\newcommand{\commentout}[1]{}

\ifx\pdftexversion\undefined
\usepackage[colorlinks,linkcolor=black,filecolor=black,citecolor=black,urlco
lor=black,pdfstartview=FitH]{hyperref}
\else
\usepackage[colorlinks,linkcolor=blue,filecolor=blue,citecolor=blue,urlcolor
=blue,pdfstartview=FitH]{hyperref}
\fi

\newcommand{\alert}[1]{\textbf{\color{red}
[[[#1]]]}\marginpar{\textbf{\color{red}**}}\typeout{ALERT:
\the\inputlineno: #1}}

\newcommand{\E}{{\mathbb{E}}}
\newcommand{\cd}{\cdot}
\newcommand{\tB}{\tilde{B}}

\newcommand{\dist}{{\rm dist}}

\newcommand{\mommit}[1]{}
\newcommand{\namedref}[2]{\hyperref[#2]{#1~\ref*{#2}}}

\newtheorem{theorem}{Theorem}
\newtheorem{lemma}{Lemma}

\newtheorem{remark}{Remark}
\newtheorem{claim}[lemma]{Claim}

\usepackage{pdfsync}
\usepackage{authblk}

\title{Path-Reporting Distance Oracles for Vertex-Labeled Graphs}

\author[1]{Ofer Neiman\thanks{Supported in part by ISF grant No. (970/21).}}
\author[1]{Alon Spector\thanks{Supported in part by the Lynn and William Frankel Center for Computer Sciences, and  ISF grant No. (970/21).}}

\affil[1]{Faculty of Computer and Information Science, Ben-Gurion University of the Negev. Emails: \texttt{neimano@cs.bgu.ac.il, alonspec@post.bgu.ac.il}}

\begin{document}

\maketitle


\begin{abstract}
Let $G=(V,E)$ be a weighted undirected graph, with $n$ vertices. 
A distance oracle is a data structure that can quickly answer distance queries, with some stretch factor. A seminal work of \cite{TZ01}, given an integer $k\ge 1$, provides such an oracle with stretch $2k-1$, query time $O(k)$, and size $O(k\cd n^{1+1/k})$. Furthermore, this oracle can also report a path in $G$ corresponding to the returned distance.

In this paper we focus on vertex-labeled graphs, in which each vertex is given a label from a set $L$ of size $\ell$. 
A {\em vertex-label distance oracle} answers queries of the form $(v,\lambda)$, where $v\in V$ and $\lambda\in L$, by reporting (an approximation to) the distance from $v$ to the closest vertex of label $\lambda$. 
Following \cite{HLWY11}, it was shown in \cite{C12} that for any integer $k> 1$, there exists a vertex-label distance oracle with stretch $4k-5$, query time $O(k)$, and size $O(k\cdot n\cdot \ell^{1/k})$. 

This state-of-the-art result suffers from two main drawbacks: The stretch is roughly a factor of 2 larger than in \cite{TZ01}, and it is not path-reporting. We address these concerns in this work, and provide the following results. 
\begin{itemize}
    \item First, we devise a {\em path-reporting} vertex-label distance oracle, at the cost of a slight increase in stretch and size. For any constant $0<\epsilon<1$, our oracle has stretch $(4k-5)\cdot(1+\epsilon)$, query time $O(k)$, and size  $O(n^{1+o(1)}\cdot \ell^{1/k})$.
    \item Second, we show how to improve the stretch to the optimal $2k-1$, at the cost of mildly increasing the query time. Specifically, we devise a vertex-label distance oracle with stretch $2k-1$, query time $O(\ell^{1/k}\cdot\log n)$, and size $O(k\cdot n\cdot \ell^{1/k})$.
\end{itemize}

\end{abstract}
\newpage
\section{Introduction}
Let $G=(V,E)$ be an undirected graph, with positive weights on the edges $w:E\rightarrow\mathbb{R}_+$. An (approximate) {\em distance oracle} is a data structure that can (approximately) answer distance queries. The oracle has stretch $t$, if for every pair $u,v\in V$, the answer it provides $\tilde{d}(u,v)$ to a given query $(u,v)$ satisfies
\[
\dist(u,v)\le \tilde{d}(u,v)\le t\cdot \dist(u,v)~,
\]
where $\dist(\cdot)$ is the distance in $G$ with respect to the weights. The main interest is the tradeoff between the stretch $t$, the size of the oracle (number of words needed to store it), and the query time. In a seminal work, \cite{TZ01} showed for any $n$-vertex graph and any integer $k\ge 1$, a distance oracle with stretch $2k-1$, size $O(kn^{1+1/k})$ and query time $O(k)$. Subsequently, the size and query time were improved to $O(n^{1+1/k})$ and $O(1)$ respectively, \cite{WN13,C15}. This state-of-the-art tradeoff is  optimal (up to constants), assuming Erd\H{o}s' girth conjecture. 

A useful feature of a distance oracle is its ability to report not just an approximation $\tilde{d}(u,v)$ to the distance, but also provide a path in $G$ between $u,v$ of length $\tilde{d}(u,v)$. An oracle with this property is called {\em path-reporting}. The original construction of \cite{TZ01} provides a path-reporting oracle, albeit the improved one of \cite{C15} does not. Recently there has been a surge of interest in path-reporting distance oracles \cite{ENW16,EP16,ACEFN20,NS24,ES23,CT24}, most of these were interested in obtaining (almost) linear size, as the \cite{TZ01} oracles always have size $\Omega(n\log n)$. 

When stating the query time of path-reporting oracles, we will follow the convention of omitting the length of the returned path. That is, if the algorithm returns a path $P$ in time $O(q+|P|)$, we will write that the query time is $O(q)$ (since reporting $P$ will always incur time $O(|P|)$).

\subsection{Vertex-Label Distance Oracles}

In this paper we focus on distance oracles for {\em vertex-labeled} graphs, introduced by \cite{HLWY11}. In this setting, there is a set $L=\{\lambda_1,\dots,\lambda_\ell\}$ of labels, and each vertex receives a single label from $L$. We would like to answer {\em vertex-label} queries, which are queries of the form $(v,\lambda)$, for $v\in V$ and $\lambda\in L$, and the goal is to return an approximation for $\dist(v,\lambda)$, which is the distance from $v$ to the nearest vertex of label $\lambda$.

Vertex-label queries arise naturally in various applications, where there are multiple vertices that can provide a service to the source vertex $v$. For instance, in a network, the labels may indicate different types of servers that provide certain functionality, or in a road network, the labels may correspond to services such as gas, food, accommodation, etc. We would like to be able to answer queries such as ``how close is the nearest gas station'', rather than the distance to a particular one. 

In \cite{HLWY11} it was shown that the distance oracles of \cite{TZ01} can be adapted to the vertex-label query setting, with stretch $4k-5$ and size  $O(kn^{1+1/k})$, regardless of the number of labels, $\ell$. However, if $\ell< n^{1/k}$, then this size is larger than the trivial solution with stretch 1 and size $O(n\cdot\ell)$, that simply stores all distances between vertices and labels.
To address this issue, Hermelin et al.~\cite{HLWY11} devised an oracle with improved size of $O(kn\cdot\ell^{1/k})$, at the cost of increasing the stretch to $2^k-1$. In \cite{C12}, the stretch was substantially improved to $4k-5$. Both results have query time $O(k)$.

The current state-of-the-art leaves two main concerns. The first is that unlike standard distance oracles, the vertex-label oracles of \cite{HLWY11,C12} of size $\approx n\cdot\ell^{1/k}$ are not {\em path-reporting}. Indeed, for many applications the ability to report paths seems appealing (for instance, rather than obtaining just the distance to the nearest gas station, we would like to receive a path to it). The second concern is that the stretch is roughly a factor of 2 away from the conjectured optimum.

\subsection{Our Results}


In this paper we provide a positive answer to both concerns mentioned above. That is, by slightly increasing the stretch, we show how to obtain path-reporting vertex-label distance oracles. In addition, we show a construction that achieves the optimal stretch of $2k-1$, at the cost of mildly increasing the query time.

Specifically, for any $n$-vertex graph, given an integer $k>1$ and constant $0<\epsilon<1$, we devise a path-reporting distance oracle, that answers vertex-label queries in $O(k)$ time, and has stretch $(4k-5)\cdot(1+\epsilon)$ and size  $O(n^{1+o(1)}\cdot \ell^{1/k})$. (Alternatively, with stretch $O(k)$ and size $\tilde{O}(n\cdot \ell^{1/k})$.) \footnote{When writing $\tilde{O}(f(n))$ we mean $O(f(n)\cdot\log^{O(1)}(f(n)))$.}

Our second result achieves for any integer $k>1$, a vertex-label distance oracle with stretch $2k-1$, size $O(k\cdot n\cdot \ell^{1/k})$, and query time $O(\ell^{1/k}\cdot\log n)$. (This oracle can also be made path-reporting, at the cost of increasing the stretch by a factor of $1+\epsilon$, and the size by a factor of $n^{o(1)}$.)

\subsection{Overview of Techniques}

Our construction of the path-reporting vertex-label distance oracle starts similarly to previous works \cite{HLWY11,C12}. Given the integer parameter $k>1$, we sample a sequence of sets $V=A_0\supseteq A_1\supseteq...\supseteq A_{k-1}\supseteq A_k=\emptyset$, as follows. For each $0\le i\le k-2$, every vertex in $A_i$ is independently sampled to $A_{i+1}$ with probability $\ell^{-1/k}$ (recall that $\ell=|L|$ is the number of labels). This sampling probability is different from the $n^{-1/k}$ used in \cite{TZ01}, which allows the size improvement from $\approx n^{1+1/k}$ to $\approx n\cdot\ell^{1/k}$. However, the main issue is that the last set $A_{k-1}$ is expected to be prohibitively large, $\frac{n}{\ell^{1-1/k}}$.

As in \cite{TZ01}, for each $v\in V$ and $0\le i\le k-1$ define {\em pivots}: $p_i(v)$ is the nearest vertex of $A_i$ to $v$, and {\em bunches}: $B_i(v)$ contains all the vertices in $A_i$ closer to $v$ than $p_{i+1}(v)$, and $B(v)=\bigcup_iB_i(v)$. The {\em cluster} of $u\in V$, $C(u)$, consists of all vertices $v$ such that $u\in B(v)$. 

\paragraph{Path-reporting construction.} The query algorithm of \cite{C12}, given the query $(v,\lambda)$, finds a pivot $p_i(v)$ such that there exists some vertex $u$ of label $\lambda$ with $p_i(v)\in B(u)$. It then returns the sum of distances from $u,v$ to this pivot $p_i(v)$. We take a similar approach, and as long as $i<k-1$, we can apply the technique of \cite{TZ01}, using clusters, to find a path from $v$ to $u$ via the cluster of $p_i(v)$. The issue arises when $i=k-1$, then the clusters of vertices in this last level $A_{k-1}$ are too large to store. To overcome this hurdle, we apply a {\em pairwise} path-reporting distance oracle, that approximately preserves distances (and reports paths), only between a given set of pairs. Indeed, we show that the set of pairs containing all the possible required connections for vertices in $A_{k-1}$ is sufficiently small, and the pairwise oracles of \cite{ES23,NS24} induce a slight increase in stretch and size.

\paragraph{Improved stretch.}
The query algorithm mentioned above, given a query $(v,\lambda)$, tests only the pivots of $v$. This {\em one-sided testing} enables stretch of $4k-5$, similarly to the context of compact routing schemes \cite{TZ01a}, where we have access to the bunches of the source $v$, but not to the bunches of the destination $u$. In order to obtain the (conjectured) optimal stretch of $2k-1$ for vertex-label queries, we need {\em two-sided testing}, as in \cite{TZ01,C15}, where we also inspect whether the pivots of the destination $u$ are in $B(v)$. Alas, the destination $u$ is not known, we are only given a label $\lambda$. 
In order to achieve two-sided testing, it seems that in the $i$-th iteration, one must check for each vertex $u$ of label $\lambda$, whether $p_i(u)\in B(v)$. Unfortunately, this may result in a linear query time. To deal with this barrier, in the pre-processing phase, for each $\lambda\in L$, we store all the level $i$ pivots of vertices with label $\lambda$, denoted $P_i(\lambda)$, in a hash table. Then, we can check if there exists a vertex of $B_i(v)$ in $P_i(\lambda)$. This will be much faster, since the size of $B_i(v)$, for $i<k-1$, is bounded by $\approx\ell^{1/k}$. As the last set of level $k-1$ can be quite large, we will need to treat it differently, which leads to an intricate case analysis when bounding the stretch.


\subsection{Related Work}

In the context of planar graphs, \cite{LMN13} showed a vertex-label distance oracle with stretch $1+\epsilon$, size $O(n\log n)$ and query time $O(\log n\cdot\log\Delta)$, where $0<\epsilon<1$ is fixed and $\Delta$ is the hop-diameter of the graph. This result is somewhat inferior to state-of-the-art distance oracles for planar graph, where \cite{LW21} obtained stretch $1+\epsilon$, size $O(n)$ and query time $O(1)$. In \cite{EFW21}, a stretch 1 vertex-label distance oracle for planar graphs is shown, with parameters (size and query) that are at most polylogarithmically larger than the standard setting.

A closely related problem is the so-called {\em colored distance oracles}, which can answer label-to-label queries. These objects were studied by \cite{KK16,HK25}, in the context of detecting patterns in texts, focusing on the $1+\epsilon$ stretch regime.

\section{Preliminaries}
    Let $G=(V,E)$ be an undirected graph with nonnegative weights on the edges $w:E\rightarrow \mathbb{R}_+$. Given a set $L$ of labels, each vertex is associated with a single label $\lambda\in L$. For $\lambda\in L$, the set of vertices of label $\lambda$ is denoted by $V_\lambda$.
	
    For vertices $u,v\in V$, define \( \text{dist}(u, v) \) as the length of the shortest path between $u,v$ with respect to the weights. If $H$ is a subgraph of $G$, we denote $\text{dist}_H(u, v)$ the distance in $H$.
	For a vertex \( v \in V \) and a set \( A \subseteq V \), define
	$\text{dist}(v, A) = \min_{u \in A} \{\text{dist}(v, u)\}$.
	Additionally, let \( \text{dist}(v, \emptyset) = \infty \).
    
	For every vertex \( v \in V \) and label \( \lambda \in L \), we define \( \lambda(v) \) as the closest node to \( v \) that has the label \( \lambda \), that is, $\text{dist}(v, \lambda(v))=\text{dist}(v, V_\lambda)$.\\

\paragraph{Pairwise distance oracles.} Given a set of pairs ${\cal P}\subseteq V\times V$, a {\em pairwise} distance oracle is required to (approximately) answer distance queries only for pairs in ${\cal P}$. A pairwise oracle is called {\em path-reporting} if it can also return a path in $G$ that achieves the approximated distance. We will use two results: the first from \cite{ES23} that has stretch close to 1, and the second from \cite{NS24} that has constant stretch, but smaller size. By plugging in $k=\log n$ in \cite[Theorem 2]{ES23} we get\footnote{Path-reporting oracles are a part of {\em interactive spanners} in \cite{ES23}.}

    \begin{theorem}[\cite{ES23}]\label{thm:prdo23}
    Given an undirected weighted $n$-vertex graph $G = (V, E)$, a parameter
$0<\epsilon \le 1$, and a set of pairs ${\cal P} \subseteq V\times V$, there exists a pairwise path-reporting distance oracle, with stretch $1+\epsilon$,
query time $O(1)$ and size
\[
O\left(|{\cal P}|\cdot \left(\frac{\log\log n}{\epsilon}\right)^{\log_{4/3}\log n}+n\log\log n\right)~.
\]
    \end{theorem}

    By plugging in a constant $c>1$ and $k=\log n$ in \cite[Theorem 13]{NS24} we get\footnote{Path-reporting oracles are a part of {\em path-reporting spanners} in \cite{NS24}.}
     \begin{theorem}[\cite{NS24}]\label{thm:prdo24}
    Given an undirected weighted $n$-vertex graph $G = (V, E)$ and a set of pairs ${\cal P} \subseteq V\times V$, there exists a pairwise path-reporting distance oracle, with stretch $O(1)$, query time $O(1)$ and size
\[
\tilde{O}\left(|{\cal P}|+n\right)~.
\]
    \end{theorem}

\section{Path-Reporting Oracles for Labeled Graphs}\label{sec:PathReporting}

    In this section we prove the following theorem.
    \begin{theorem}
        For any $n$-vertex graph $G=(V,E)$ labeled by a set $L$ of $\ell$ labels, any integer $k> 1$ and constant $0<\epsilon<1$, there exists a path-reporting distance oracle that can answer any vertex-label query in time $O(k)$, with stretch $(4k-5)\cdot(1+\epsilon)$ and size $O(n^{1+o(1)}\cdot \ell^{1/k})$. (Alternatively, with stretch $O(k)$ and size $\tilde{O}(n\cdot \ell^{1/k})$.)
    \end{theorem}

    \begin{remark}
    The parameter $\epsilon$ can be nonconstant, in which case the size increases by a factor of $\epsilon^{-O(\log\log n)}$.
    The term $n^{o(1)}$ is in fact only sightly super-polylogarithmic: $(\log\log n)^{O(\log\log n)}$. 
    \end{remark}

\subsection{Preprocessing}
Define a random sequence of sets: $V=A_0\supseteq A_1\supseteq\dots\supseteq A_k=\emptyset$, where for each $0\le i\le k-2$, every vertex of $A_i$ is selected for $A_{i+1}$ independently with probability $\ell^{-1/k}$. For each vertex $v\in V$ and $0\le i\le k-1$ define the $i$-th \emph{pivot} $p_i(v)\in A_i$ as the nearest vertex\footnote{\label{foot:note}
If there are ties, we use a consistent tie-breaking scheme: first, we prefer to take the vertex belonging to the highest set. If 
$\text{dist}(v,A_i) = \text{dist}(v, A_{i+1})$, then $p_i(v) \leftarrow p_{i+1}(v)$; 
Otherwise, we choose arbitrarily from the set 
$\{ u \in A_i \mid \text{dist}(v,A_i) = \text{dist}(v, u) \}$.
} in $A_i$ to $v$, i.e.
\[\dist(v,A_i)=\dist(v,p_i(v))~.\]

Define the $i$-th {\em bunch} of $v\in V$ by
	\[
	B_i(v) = \{ u \in A_i  ~:~ \text{dist}(v, u) < \text{dist}(v, A_{i+1}) \}~.
	\]
Note that the distance to $A_k=\emptyset$ is $\infty$, so that $B_{k-1}(v)=A_{k-1}$ for every $v\in V$. Let
\[
	B(v) = \bigcup_{i=0}^{k-1} B_i(v)~.
\]

Additionally, we define \( \widetilde{B}(v) \) as the bunch excluding the last level, that is,
	\[
	\widetilde{B}(v) = \bigcup_{i=0}^{k-2} B_i(v)=B(v)\setminus A_{k-1}.
	\]

The definition of bunches extends naturally to labels, by defining for any $\lambda\in L$,
\[
B(\lambda)=\bigcup_{v\in V_\lambda}B(v)~,
\]
and similarly define $\widetilde{B}(\lambda)=B(\lambda)\setminus A_{k-1}$. For every $u\in \widetilde{B}(\lambda)$, let 
\[\lambda_B(u)=\text{argmin}_{x\in V_\lambda: u\in \widetilde{B}(x)}\{\dist(u,x)\}\]
in other words, this is the closest vertex $x$ to $u$ that caused $u$ to enter $\widetilde{B}(\lambda)$; so $x$ has label $\lambda$ and also satisfies $u\in \widetilde{B}(x)$. 

\paragraph{Clusters.} Similarly to \cite{TZ01}, we use the ``inverses'' of bunches, called {\em clusters}. For any $0\le i\le k-1$ and $u\in A_i\setminus A_{i+1}$, define
\[
C(u)=\{v\in V~:~ \text{dist}(u,v)<\text{dist}(v,A_{i+1})\}~.
\]
Note that $v\in C(u)$ iff $u\in B(v)$. It was shown in \cite{TZ01} that one can efficiently compute a tree $T(u)$ rooted at $u$ that spans $C(u)$, so that $\text{dist}_{T(u)}(u,v)=\text{dist}(u,v)$ for every $v\in C(u)$.

As pointed out in the introduction, the last set $A_{k-1}$ is expected to be excessively large, of size $\frac{n}{\ell^{1-1/k}}$. So unlike \cite{TZ01} we cannot use clusters to derive paths for vertices in that set.
Indeed, the main difference between previous works and our construction, is that we employ a path-reporting pairwise distance oracle for vertices in the last set $A_{k-1}$. Specifically, we define a set of required pairs ${\cal P}$ as follows: For every $v\in V$ we add the pair $(v,p_{k-1}(v))$ to ${\cal P}$, and for every $u\in A_{k-1}$ and every label $\lambda\in L$, we add the pair $(u,\lambda(u))$ to ${\cal P}$ (recalling that $\lambda(u)$ is the nearest vertex to $u$ that has label $\lambda$).
That is, define \[\mathcal{P} = \{(v, p_{k-1}(v)) : v \in V\} \cup\{(u, \lambda(u)) : u \in A_{k-1}, \; \lambda \in L \} ~,\]
and let $O_{\cal P}$ be the pairwise path-reporting distance oracle from Theorem~\ref{thm:prdo23} (respectively, from Theorem~\ref{thm:prdo24}) for the set of pairs ${\cal P}$.

The following claims will be useful:
\begin{claim}\label{claim:v in B}
		For all $v\in V$, and for all $0\leq i\leq k-1$, $p_i(v)\in B(v)$.
	\end{claim}
	\begin{proof}
		We prove the claim by induction on \( i \), proceeding downward. The base case is for \( i = k - 1 \), which holds since \( p_{k-1}(v) \in A_{k-1} \subseteq B(v) \) for every \( v \in V \).
		Now, assume the claim holds for \( i+1 \le k-1 \), so that \( p_{i+1}(v) \in B(v) \). If $\dist(v,A_i)=\dist(v,A_{i+1})$, then as we chose pivot by a consistent tie breaking scheme, preferring vertices from higher sets, it follows that \( p_i(v) = p_{i+1}(v) \in B(v) \). Otherwise, \( \dist(v,p_i(v)) = \dist(v,A_i) < \dist(v,A_{i+1}) \), which implies \( p_i(v) \in B_i(v)\subseteq B(v)\).
	\end{proof}
Next, we record some known facts about the expected size of the sets $A_i$ and of the bunches $\widetilde{B}(v)$. For $0\le i\le k-1$, to be in $A_i$, a vertex needs to be sampled $i$ times. Since each sample is done independently with probability $\ell^{-1/k}$, we have that
\begin{equation}\label{eq:ai}
\E[|A_i|]=\frac{n}{\ell^{i/k}}~.
\end{equation}

\begin{restatable}{claim}{bunchsize}
\label{claim:bsize}
For all $v\in V$,  $\E[|\tB(v)|]\leq(k-1)\cdot\ell^{1/k}$. Furthermore, with high probability it holds that for all $v\in V$, $|\tB(v)|=O(\ell^{1/k}\cdot\log n)$.
\end{restatable}

The proof is similar to \cite{TZ01,C12}, and is deferred to Appendix~\ref{app:miss}.

 \paragraph{Data structure.} Our path-reporting distance oracle will store in memory the following items:
\begin{enumerate}
\item For every $v\in V$, store in a hash table $\widetilde{B}(v)$, and $p_i(v)$ for all $0\le i\le k-1$. Each vertex is stored along with its distance to $v$.
\item For every $u\in V\setminus A_{k-1}$, store $C(u)$ and the corresponding tree $T(u)$.
\item For every label $\lambda\in L$, store in a hash table $\widetilde{B}(\lambda)$. Every $u\in\widetilde{B}(\lambda)$ is stored with $\lambda_B(u)$
.
\item The collection of pairs ${\cal P}$, and the pairwise distance oracle $O_{\cal P}$.
\end{enumerate}

\subsection{Size Analysis}\label{Size Analysis}

Note that every vertex has $k$ pivots, so by Claim~\ref{claim:bsize}, the total storage for the first item of our data structure, the bunches and pivots, is expected to be of $O(k\cdot n\cdot\ell^{1/k})$ words.

In order to bound the size of clusters stored, we recall that clusters are inverses of bunches, i.e., $C(u)=\{v\in V: u\in B(v)\}$. Furthermore, whenever $u \in V \setminus A_{k-1}$ and $u \in B(v)$, we have that $u\in\widetilde{B}(v)$, so 
\begin{equation}\label{eq:cluster}
    \sum_{u\in V\setminus A_{k-1}}|C(u)|=\sum_{u\in V\setminus A_{k-1}}\sum_{v\in V:u\in \widetilde{B}(v)}1=\sum_{v\in V}|\widetilde{B}(v)|~,
\end{equation}

So the total storage required for the second item, the clusters of vertices outside $A_{k-1}$, is expected to be  $\E\left[\sum_{v\in V}|\widetilde{B}(v)|\right]=O(k\cdot n\cdot\ell^{1/k})$
as well.

Since $\{V_\lambda\}_{\lambda\in L}$ is a partition of $V$, every bunch $\widetilde{B}(v)$ is counted exactly once in one of the $\widetilde{B}(\lambda)$, so the expected size of third item is also bounded by $O(k\cdot n\cdot\ell^{1/k})$ words.

Finally, the size of the pairwise distance oracle $O_{\cal P}$ is at most:
\[O\left(|{\cal P}|\cdot \left(\frac{\log\log n}{\epsilon}\right)^{\log_{4/3}\log n}+n\log\log n\right)\]
if we use Theorem~\ref{thm:prdo23}, and $\tilde{O}\left(|{\cal P}|+n\right)$ if we use Theorem~\ref{thm:prdo24}. It remains to bound the expected size of ${\cal P}$. The first set of pairs, from each vertex to its pivot in $A_{k-1}$, consists of $n$ pairs, so we focus now on the second set of pairs in ${\cal P}$:  $\{(u, \lambda(u)) : u \in A_{k-1}, \; \lambda \in L \}$. By \eqref{eq:ai}, we expect $\frac{n}{\ell^{1-1/k}}$ vertices to be in $A_{k-1}$, each of them adds $\ell$ pairs to ${\cal P}$, one for each label. We get that
\[
\E[|{\cal P|}]=n+\frac{n}{\ell^{1-1/k}}\cdot\ell=O(n\cdot\ell^{1/k})~.
\]
Thus, the expected size of the pairwise distance oracle $O_{\cal P}$ of  Theorem~\ref{thm:prdo23} is at most
\[
O\left( n\cdot\ell^{1/k}\cdot\left(\frac{\log\log n}{\epsilon}\right)^{\log_{4/3}\log n}\right)=O\left( n^{1+o(1)}\cdot\ell^{1/k}\right)~,
\]
for constant $0<\epsilon<1$. Alternatively, it is at most $\tilde{O}(n\cdot\ell^{1/k})$ using Theorem~\ref{thm:prdo24}.

\subsection{Query Algorithm}\label{sec:query}

Let $(v,\lambda)$ be the query, for some $v\in V$ and $\lambda\in L$. The query algorithm consists of two phases. 


\begin{description}
\item [Phase 1:] For every $0\le i\le k-2$, check whether $p_i(v)\in \widetilde{B}(\lambda)$. If so, then let $x_i=\lambda_B(p_i(v))$ be the vertex in $V_\lambda$ that is returned by our data structure when querying $\widetilde{B}(\lambda)$ (recall that this is the closest vertex of label $\lambda$ to $p_i(v)$ such that $p_i(v)\in \widetilde{B}(x_i)$). Define $D_i=\dist(v,p_i(v))+\dist(p_i(v),x_i)$.  Note that both of these distances are stored in the first item of our data structure.

In the case that $p_i(v)\notin \widetilde{B}(\lambda)$, then simply set $D_i=\infty$.

\item [Phase 2:] Extract $x_{k-1}=\lambda(p_{k-1}(v))$ from ${\cal P}$ (recall that this is the closest vertex of label $\lambda$ to $p_{k-1}(v)$), and use the pairwise path-reporting distance oracle $O_{\cal P}$ to (approximately) compute $D_{k-1}=\dist(v,p_{k-1}(v))+\dist(p_{k-1}(v),x_{k-1})$.
\end{description}

\paragraph{Distance.} Let $i^*=\text{argmin}_{0\le i\le k-1}\{D_i\}$, and return $D_{i^*}$ as the approximation to the distance from $v$ to $\lambda(v)$. 

\paragraph{Path-reporting.} If $i^*<k-1$, then let $T=T(p_{i^*}(v))$. Note that $x_{i^*}\in C(p_{i^*}(v))$ and by Claim~\ref{claim:v in B} also $v\in C(p_{i^*}(v))$, so both $x_{i^*}$ and $v$ are in $T$, and we return the path in $T$ from $x_{i^*}$ to $v$ (by going from each of them to the root $p_{i^*}(v)$). Otherwise, if $i^*=k-1$, use the path-reporting oracle $O_{\cal P}$ to return the path from $v$ to $p_{k-1}(v)$ concatenated with the path from $p_{k-1}(v)$ to $x_{k-1}$.

\paragraph{Query time.} Let $P$ be the returned path, then the running time of the query algorithm is $O(k+|P|)$, since in each of the $k$ iterations we check $O(1)$ vertices and distances from the hash tables. If $i^*<k-1$, then the path $P$ in the tree $T$ can be recovered in $O(|P|)$ time.  Otherwise, the oracle $O_{\cal P}$ recovers $P$, also in $O(|P|)$ time. As we state the query time omitting the length of the returned path, we indeed obtain query time $O(k)$.

\subsection{Stretch Analysis}\label{stretch}

In what follows we assume that $A_{k-1}\neq\emptyset$ (which happens w.h.p.), still this is without loss of generality, as otherwise we can simply redefine $k$ as the minimal for which $A_k=\emptyset$.
The following lemma is very similar to \cite[Lemma A.1]{TZ01a}, we give a proof for completeness in Appendix~\ref{app:miss}.

\begin{restatable}{lemma}{fourk}\label{lem:4k-3}
		Fix any $u,v\in V$, and let \(0\le j\le k-1\) be the smallest index such that \(p_j(v)\in B(u)\). Then,
		\[
        \dist(v,p_j(v))+\dist(u,p_j(v))\leq (4k-3)\cd \dist(u,v)
        \]
	\end{restatable}

Let $(v,\lambda)$ be the given query, set $u=\lambda(v)$ the closest vertex to $v$ that has label $\lambda$. We are now ready to analyze the stretch of the returned path, according to the following two cases. We stress that these cases are not analogous to the two phases of the algorithm, since in the analysis here we consider $\widetilde{B}(u)$, and $u$ is not known to the query algorithm (which considers $\widetilde{B}(\lambda)$).
\begin{description}
\item [Case 1:] There exists $0\le i\le k-2$ such that $p_i(v)\in \widetilde{B}(u)$. Denote $i$ as the minimal such index, and let $x_i=\lambda_B(p_i(v))$. Recall that $x_i$ is the closest vertex to $p_i(v)$ that satisfies: 1) has label $\lambda$ , and 2) contains $p_i(v)$ in its bunch $\widetilde{B}(\cdot)$. Since $u$ satisfies both these conditions, we have that
\begin{equation}\label{eq:pixi}
\dist(p_i(v),x_i)\le\dist(p_i(v),u)~.
\end{equation}
Using that $i$ is the minimal such that $p_i(v)\in B(u)$, by Lemma~\ref{lem:4k-3} it follows that 
\[
D_i=\dist(v,p_i(v))+\dist(p_i(v),x_i)\stackrel{\eqref{eq:pixi}}{\le}\dist(v,p_i(v))+\dist(p_i(v),u)\le (4k-3)\cdot \dist(v,u)~.
\]

So by definition, the returned distance $D_{i^*}\le D_i$ is also a $4k-3$ approximation. Note that the returned path in the tree $T=T(p_{i^*}(v))$ incurs no additional stretch.

\item [Case 2:] For all $0\le i\le k-2$, $p_i(v)\notin \widetilde{B}(u)$. Denote $i$ as the minimal index such that $p_i(v)\in B(u)$. (There is such $0\le i\le k-1$ because $B_{k-1}(u)= A_{k-1}\subseteq B(u)$.) As
$p_i(v)\in B(u) \setminus \widetilde{B}(u)$ it follows that $p_i(v)\in A_{k-1}$ i.e., $p_i(v)=p_{k-1}(v)$ (see footnote~\ref{foot:note}).
 Recall that in phase 2 the algorithm sets $x_{k-1}=\lambda(p_{k-1}(v))$, the closest vertex of label $\lambda$ to $p_{k-1}(v)$, so we have that 
\begin{equation}\label{eq:pixis}
\dist(p_{k-1}(v),x_{k-1})\le\dist(p_{k-1}(v),u)~.
\end{equation}
Once again using Lemma~\ref{lem:4k-3} we get
\begin{eqnarray*}
			D_{k-1} &=& \dist(v,p_{k-1}(v))+\dist(p_{k-1}(v),x_{k-1})\\
			&\stackrel{\eqref{eq:pixis}}{\le}& \dist(v,p_{k-1}(v))+\dist(p_{k-1}(v),u)\\
			&\le&(4k-3)\cdot\dist(v,u)~.
\end{eqnarray*}
So in this case as well, the distance $D_{i^*}\le D_{k-1}$ is also a $4k-3$ approximation. However, in this case the returned path is using the path-reporting pairwise distance oracle $O_{\cal P}$, which incurs an additional $1+\epsilon$ factor to the stretch if we used Theorem~\ref{thm:prdo23} for its construction, or an $O(1)$ factor if we used Theorem~\ref{thm:prdo24}.

\end{description}

\subsubsection{Improving the Stretch}

To obtain the slightly improved stretch $4k-5$, we take a similar approach to \cite{TZ01a}. The basic idea is to make a two-sided test in the first level, and one-sided tests in all other levels.\footnote{In the context \cite{TZ01a} were interested in, compact routing schemes, this amounts to storing the clusters of vertices in the first level $A_0\setminus A_1$, and ensuring they are sufficiently small.} To implement this idea in our setting, we define the cluster of a label $\lambda\in L$ by
\[
C(\lambda)=\bigcup_{u\in V_\lambda\setminus A_1}C(u)~.
\]
Note that we take a union only over first-level vertices in $A_0\setminus A_1$. Therefore, by \eqref{eq:cluster}, we can store every $C(\lambda)$ in a hash table with expected space only $O(k\cdot n\cdot \ell^{1/k})$. Every vertex $v\in C(\lambda)$ is stored together with $u=\lambda_C(v)$, which is the closest vertex to $v$ that satisfies $u\in V_\lambda\setminus A_1$ and $v\in C(u)$. In other words, this is the closest vertex to $v$ that caused $v$ to enter $C(\lambda)$.

The query algorithm will start by testing whether $v\in C(\lambda)$, and if so, it will return the path from $u=\lambda_C(v)$ to $v$ in $T(u)$ (note that as $u\notin A_{k-1}$, this tree is indeed stored in our data structure). If $v\notin C(\lambda)$, continue with the query algorithm described in Section~\ref{sec:query}. The following claim ensures that we get stretch 1 if the query algorithm finds $v\in C(\lambda)$.
\begin{claim}
If $v\in C(\lambda)$ and $u=\lambda_C(v)$, then $u=\lambda(v)$.
\end{claim}
\begin{proof}
Since $u\in A_0\setminus A_1$ and $v\in C(u)$, by definition it follows that $u\in B_0(v)$. Seeking contradiction, suppose that $u'=\lambda(v)$ is different than $u$. If $u'\in A_1$, then by definition of $B_0$, and of a pivot, we get that
\[
\dist(v,u)<\dist(v,p_1(v))\le\dist(v,u')~,
\]
contradiction. Otherwise, $u'\in A_0\setminus A_1$, but then $u'\in V_\lambda\setminus A_1$ is closer to $v$ than $u$, so $u'\in B_0(v)$ as well, contradicting the definition of $u=\lambda_C(v)$ as the closest such vertex.
\end{proof}
Next, we show an improved bound on the distance from $v$ to its first pivot in the case that $v\notin C(\lambda)$.
\begin{claim}\label{claim:C1}
Suppose that $v\notin C(\lambda)$, then $\dist(v,p_1(v))\le\dist(v,\lambda(v))$.
\end{claim}
\begin{proof}
Let $u=\lambda(v)$. If $u\in A_1$, then by definition of a pivot, we have that $\dist(v,u)\ge \dist(v,p_1(v))$. Otherwise, $u\in A_0\setminus A_1$, but by the assertion of the lemma, $v\notin C(u)$, so $u\notin B_0(v)$, which implies that $\dist(v,u)\ge \dist(v,p_1(v))$.
\end{proof}

Recall that in the stretch analysis, \eqref{eq:vpi} for $i=1$ gave us the bound $\dist(v,p_1(v))\le 2\Delta$ (where $\Delta=\dist(v,\lambda(v))$). Now, with the improved bound of $\dist(v,p_1(v))\le \Delta$ from Claim~\ref{claim:C1}, the stretch will be smaller by an additive 2 (see \cite[Lemma A.2]{TZ01a}), resulting in stretch $4k-5$ instead of $4k-3$.

\section{Distance Oracles for Labeled Graphs with $2k-1$ Stretch}
	In this section we prove the following theorem.
	\begin{theorem}\label{thm:2k-1}
      For any $n$-vertex graph $G=(V,E)$ labeled by a set $L$ of $\ell$ labels, and any integer $k\ge 1$, there exists a distance oracle that can answer any vertex-label query in time \( O(\ell^{1/k}\cdot\log n) \), with stretch $2k-1$ and size $O(k\cdot n \cd \ell^{1/k})$. 
	\end{theorem}

\begin{remark}
For any $0<\epsilon<1$, the oracle of Theorem~\ref{thm:2k-1} can be made path-reporting as well, at the cost of increasing the stretch to $(2k-1)\cdot(1+\epsilon)$, and the size to 
\[O\left(\left(\frac{\log\log n}{\epsilon}\right)^{O(\log\log n)}\cdot n \cd \ell^{1/k}\right)\]
\end{remark}

	\subsection{Preprocessing}
	We will use a similar data structure as in Section~\ref{sec:PathReporting}, excluding the pairwise distance oracle $O_{\cal P}$ and the clusters. Additionally, for every $0\leq i\leq k-2$ and every $\lambda\in L$, we will store a set containing all the $i$-th level pivots of vertices with label $\lambda$. Formally,
	\[P_i (\lambda)=\{p_i (v): v\in V_\lambda\}  \]
	For every
	 $y\in P_i(\lambda)$,
	  let $\lambda_{P_i} (y)=\text{argmin}_{x\in V_\lambda: p_i(x)=y} \{\dist(x,y)\}$,
	  	 in other words, this is the closest vertex $x$ to $y$ that caused $y$ to enter $P_i (\lambda)$; so $x$ has label $\lambda$ and also satisfies $p_i(x)=y$.

We will also store for all vertices in the last level $A_{k-1}$, the distances to the closest vertex of label $\lambda$, for every $\lambda\in L$.
         
\paragraph{Data structure.} Our distance oracle will store in memory the following items:
	\begin{enumerate}
		\item For every $v\in V$, store in a hash table $\widetilde{B}(v)$, and $p_i(v)$ for all $0\le i\le k-1$. Each vertex is stored along with its distance to $v$.
		\item For every label $\lambda\in L$, store in a hash table $\widetilde{B}(\lambda)$. Every $u\in\widetilde{B}(\lambda)$ is stored with $\lambda_B(u)$ and the distance between them.
		\item \label{item:forth} For every label $\lambda\in L$ and for every $0\leq i\leq k-2$ , store in a hash table $P_i (\lambda)$.
		 Every $y\in P_i (\lambda)$ is stored with $\lambda_{P_i} (y)$ and the distance between them.
         \item  For every vertex $v\in A_{k-1}$ and every label $\lambda\in L$, store $\lambda(v)$ and the distance $\dist(v,\lambda(v))$.
	\end{enumerate}

	\subsection{Size Analysis}
	In section \ref{Size Analysis} it was shown that the total expected storage required for the first two items is  $O(k\cdot n\cdot\ell^{1/k})$ words. We show next that the third item requires only $O(k\cdot n)$ space. To see this, note that for any $0\leq i\leq k-2$ and $\lambda\in L$, it holds that $|P_i(\lambda)| \leq |V_\lambda|$ (since every vertex in $V_\lambda$ adds one of its pivots to this set). Therefore, \[\sum_{i=0}^{k-2}\sum_{\lambda\in L}|P_i(\lambda)| \leq k\cdot \sum_{\lambda\in L}|V_\lambda|=k\cdot n~.\]
    For the last item, recall that by \eqref{eq:ai} we have $\E[|A_{k-1}|]=\frac{n}{\ell^{1-1/k}}$. As every vertex in $A_{k-1}$ stores $O(\ell)$ words, the expected space required for the last item is $O(n\cdot\ell^{1/k})$.
    
    Thus, the total expected size of our data structure is $O(k\cdot n\cdot\ell^{1/k})$ words.
    
	\subsection{Query Algorithm}
	
	Let $(v,\lambda)$ be the query, for some $v\in V$ and $\lambda\in L$. The query algorithm consists of three phases. 
	Denote $\tB(v)=\{v_1,v_2,\dots,v_{|\tB(v)|}\}$.
	
	\begin{description}
		\item [Phase 1:] For every $0\le i\le k-2$, check whether $p_i(v)\in \widetilde{B}(\lambda)$. If so, then let $x_i=\lambda_B(p_i(v))$ be the vertex in $V_\lambda$ that is returned by our data structure when querying $\widetilde{B}(\lambda)$ (recall that this is the closest vertex of label $\lambda$ to $p_i(v)$ such that $p_i(v)\in \widetilde{B}(x_i)$). Define $D_i=\dist(v,p_i(v))+\dist(p_i(v),x_i)$. Note that both of these distances are stored in the first item of the data structure.
		
		In the case that $p_i(v)\notin \widetilde{B}(\lambda)$, then simply set $D_i=\infty$.
		
		\item [Phase 2:] For every $1\leq j \leq |\tB(v)|$, let $i=i(j)$ be the unique index such that $v_j\in A_i\setminus A_{i+1}$, noting that $0\le i\le k-2$ because $\tB(v)$ does not contain vertices in $A_{k-1}$. Then, check whether $v_j\in P_i(\lambda)$. If so, then let $y_j=\lambda_{P_i} (v_j)$ be the vertex in $V_\lambda$ that is returned by our data structure when querying $P_i(\lambda)$ (recall that this is the closest vertex of label $\lambda$ to $v_j$ such that $v_j=p_i(y_j)$). Define $E_j=\dist(v,v_j)+\dist(v_j,y_j)$. Again, both of these distances are stored in the first item of the data structure.
        
		In the case that $v_j\notin P_i(\lambda)$, then simply set $E_j=\infty$.
		
		\item [Phase 3:] Extract $x_{k-1}=\lambda(p_{k-1}(v))$ from the data structure (recall that this is the closest vertex of label $\lambda$ to $p_{k-1}(v)$). Compute $D_{k-1}=\dist(v,p_{k-1}(v))+\dist(p_{k-1}(v),x_{k-1})$. These distances are stored in the first (respectively, fourth) item of the data structure. 
        
	\end{description}
	
\paragraph{Distance.} Let $i^*=\text{argmin}_{0\le i\le k-1}\{D_i\}$, and $j^*=\text{argmin}_{1\le j\le |\tB(v)|}\{E_j\}$. 
   We return $\min\{D_{i^*},E_{j^*}\}$ as the approximation to the distance from $v$ to $\lambda(v)$.

\paragraph{Query time.} The running time of the query algorithm is $O(k)$ for phase 1, $O(|\tB(v)|)$ for phase 2, and $O(1)$ for phase 3, so in total it is $O(k+|\tB(v)|)$.
	By Claim~\ref{claim:bsize}, with high probability for every $v\in V$ we have $|\tB(v)|\leq O(\ell ^{1/k}\cdot\log n)$. As $k\le\log\ell\le\log n$, we obtain query time $O(\ell ^{1/k}\cdot\log n)$.

\subsubsection{Path-Reporting}

The query algorithm described above can be converted to a path-reporting one, by incorporating the technique of Section~\ref{sec:PathReporting}. To this end, we will add to the data structure the following items,  for a given a parameter $0<\epsilon<1$.
\begin{itemize}
\item For every $u\in V\setminus A_{k-1}$, store $C(u)$ and the corresponding tree $T(u)$.
\item The collection of pairs ${\cal P}$, and the pairwise distance oracle $O_{\cal P}$.
\end{itemize}

As proved in Section~\ref{Size Analysis}, this will increase the size of our data structure by a factor of $\left(\frac{\log\log n}{\epsilon}\right)^{O(\log\log n)}$. Note that the distance returned in phase 1 or 2 of the query algorithm is the sum of distances from a pivot $p\in V\setminus A_{k-1}$, to two vertices in $C(p)$. Therefore, the tree $T(p)$ which is stored in our data structure, contains the corresponding paths, and they can be recovered in linear time in the length of the paths. Finally, if the distance returned is from phase 3, we can use $O_{\cal P}$ to report the paths, increasing the stretch by a factor of $1+\epsilon$.

\subsection{Stretch Analysis}
	
As in section~\ref{stretch}, we assume that $A_{k-1}\neq\emptyset$.
The following lemma is similar to \cite[Lemma 3.3]{TZ01}, we give a proof for completeness in Appendix~\ref{app:miss}.

\begin{restatable}{lemma}{twok}\label{lem:2k-1}
    Fix any $u,v\in V$, and let $0\le i\le k-1$ be the smallest index such that \(p_i(v)\in B(u)\) or \(p_i(u)\in B(v))\) . Then:
    \[dist(v,p_i(v))+dist(p_i(v),u)\leq (2i+1)\cd dist(u,v)\]
    \end{restatable}

	Let $(v,\lambda)$ be the given query, set $u=\lambda(v)$ the closest vertex to $v$ that has label $\lambda$. Denote $i_u=\text{argmin}_{0 \leq i \leq k-1}\{p_i(u)\in B(v)\}$ as the minimal index $i$ such that $B(v)$ contains the $i$-th pivot of $u$ (note that $i_u$ is well defined, since $p_{k-1}(u)\in A_{k-1}\subseteq B(v)$). Likewise, let $i_v=\text{argmin}_{0 \leq i \leq k-1}\{p_i(v)\in B(u)\}$. We are now ready to analyze the stretch of the returned distance, according to the following cases. As before, these cases are not analogous to the phases of the algorithm.

	\begin{description}
		\item [Case 1:] $i_v\leq i_u$ and $p_{i_v}(v)\notin A_{k-1}$. Let $i=i_v$, and note that
		$p_i(v)\in \tB(u)$.
		In phase 1. of the query, we take $x_i$ which is the closest vertex to $p_i(v)$ that satisfies: 1) has label $\lambda$ , and 2) contains $p_i(v)$ in its bunch $\widetilde{B}(\cdot)$. Since $u$ satisfies both these conditions, we have that
		\begin{equation}\label{eq:pixi2}
			\dist(p_i(v),x_i)\le\dist(p_i(v),u)~.
		\end{equation}
		By the minimality of $i=i_v\le i_u$, 
		we may use Lemma~\ref{lem:2k-1} and get that 
		\[
		D_i=\dist(v,p_i(v))+\dist(p_i(v),x_i)\stackrel{\eqref{eq:pixi2}}{\le}\dist(v,p_i(v))+\dist(p_i(v),u)\le (2i+1)\cdot \dist(v,u)~,
		\]
        and as $i\le k-1$ we conclude that
        \[
        D_i\le (2k-1)\cd\dist(v,u)~.
        \]
        
		\item [Case 2:] $i_u< i_v$ and $p_{i_u}(u)\notin A_{k-1}$.
        Let $i=i_u$, and note that there exists $v_j\in \tB(v)$ such that $v_j=p_i(u)$. Let $i'$ be the unique index such that $v_j\in A_{i'}\setminus A_{i'+1}$ (observe that $i\le i'\le k-2$, since $v_j\in A_i$ but $v_j\notin A_{k-1}$). In phase 2. of the query we take $y_j=\lambda_{P_{i'}}(v_j)$, which is the closest vertex to $v_j$ that satisfies: 1) has label $\lambda$ , and 2) $v_j$ is its level $i'$ pivot. We would like to show that $u$ satisfies both these conditions: clearly it has label $\lambda$, and by our consistent choice of pivots, since $v_j$ is the level $i$ pivot of $u$, it is also its level $i'$ pivot (see Footnote~\footref{foot:note}). 
        We conclude that
		\begin{equation}\label{eq:pixi3}
			\dist(v_j,y_j)\le\dist(v_j,u)~.
		\end{equation}

		By the minimality of $i=i_u<i_v$,  
		we may use Lemma~\ref{lem:2k-1} (switching the roles of $u,v$). Recalling that $v_j=p_i(u)$, we get that 
		\[
		E_j=\dist(v,v_j)+\dist(v_j,y_j)\stackrel{\eqref{eq:pixi3}}{\le} \dist(v,v_j)+\dist(v_j,u)\le (2k-1)\cdot \dist(v,u)~.
		\]

		\item [Case 3:] $i_v\leq i_u$ and $p_{i_v}(v)\in A_{k-1}$. 
        Let $i=i_v$, and note that by our consistent choice of pivots, $p_i(v)=p_{k-1}(v)$.
		Recall that in phase 3., the algorithm sets $x_{k-1}=\lambda(p_{k-1}(v))$, the closest vertex of label $\lambda$ to $p_{k-1}(v)$, so we have that
		\begin{equation}\label{eq:pixis4}
			\dist(p_{k-1}(v),x_{k-1})\le\dist(p_{k-1}(v),u)~.
		\end{equation}
		Again using the minimality of $i=i_v$, by Lemma~\ref{lem:2k-1} we get
		\begin{eqnarray*}
		    D_{k-1}&=&\dist(v,p_{k-1}(v))+\dist(p_{k-1}(v),x_{k-1})\\&\stackrel{\eqref{eq:pixis4}}{\le}& \dist(v,p_{k-1}(v))+\dist(p_{k-1}(v),u)\\		
		&=& \dist(v,p_i(v))+\dist(p_i(v),u)\\
		&\le& (2k-1)\cdot\dist(v,u)~.
		\end{eqnarray*}
		
		\item [Case 4:] $i_u< i_v$ and $p_{i_u}(u)\in A_{k-1}$. Set $i=i_u$ and again by the consistent choice of pivots, $p_i(u)=p_{k-1}(u)$.
		By definition of a pivot,
		\begin{equation}\label{eq:pixis5}
			\dist(v,p_{k-1}(v))\le\dist(v,p_{k-1}(u))=\dist(v,p_i(u))~.
		\end{equation}
        Using the triangle inequality, and the minimality of $i$, we have
		\begin{eqnarray*}
			D_{k-1} &=& \dist(v,p_{k-1}(v))+\dist(p_{k-1}(v),x_{k-1})\\
			&\stackrel{\eqref{eq:pixis4}}{\le}& \dist(v,p_{k-1}(v))+\dist(p_{k-1}(v),u)\\
			&\le&\dist(v,p_{k-1}(v))+\dist(p_{k-1}(v),v)+\dist(v,u)\\
			&\stackrel{\eqref{eq:pixis5}}{\le} & 2\cd \dist(v,p_i(u))+\dist(v,u)\\
			&\le&\dist(v,p_i(u))+\dist(u,p_i(u))+2\cd\dist(v,u)\\
			&\stackrel{\text{Lemma~\ref{lem:2k-1}}}{\le}&(2\cd i+3)\cdot\dist(v,u)~.
\end{eqnarray*}
Finally, note that as $i<i_v\le k-1$, it follows that $i\le k-2$, therefore 
\[
D_{k-1}\leq(2\cd (k-2)+3)\cd\dist(v,u) \\
			= (2k-1)\cdot\dist(v,u)~.
\]

	\end{description}
	In each of these four cases we demonstrated that one of the $D_i$ or $E_j$ is a $2k-1$ approximation to $\dist(v,u)$, so we conclude the same for returned distance $\min\{D_{i^*},E_{j^*}\}$.

	\bibliography{references}

\appendix
\section{Missing Proofs}\label{app:miss}

\subsection{Proof of Claim~\ref{claim:bsize}}

\bunchsize*

\begin{proof}
    Fix any $0\le i\le k-2$. As shown in \cite{TZ01}, $|B_i(v)|$ is stochastically dominated by a geomteric random variable $X_i$ with parameter $p=\ell^{-1/k}$. It follows that 
    \[\E[|\tB(v)|]=\sum_{i=0}^{k-2}\E[|B_i(v)|]\leq\sum_{i=0}^{k-2}\E[X_i]=\sum_{i=0}^{k-2}p^{-1}=(k-1)\cdot \ell^{1/k}~.\]

    Furthermore, the variables $X_i$ are independent, which implies that for any $s>0$, if $Y$ is a binomial random variable with parameters $s$ and $p$, then
    \[
	\Pr[|\tB(v)| > s] = \Pr\left[\sum_{i=0}^{k-2}|B_i(v)| > s\right] \leq \Pr\left[\sum_{i=0}^{k-2} X_i > s\right]=\Pr[Y<k].
	\]
The final equality was observed in \cite{TZ01} (in the first $s$ attempts there were at most $k-1$ successes). Setting \( s = 16\cdot \ell^{1/k} \cdot \ln n \), let $\mu =\E[Y] = s \cdot p=16\ln n$, and apply a Chernoff bound,
	
	\[
	\Pr[Y < (1-\delta) \mu] < \exp(-\mu \delta^2 /2)~,
	\]
	
	 with \( \delta = 1/2 \). Noticing that \( k \leq \log\ell< 8 \cdot \ln n \), it follows that
	
	\[
	\Pr[Y < k] \leq \Pr[Y < 8 \cdot \ln n] \leq \exp \left(-\frac{\mu}{8}\right) = \frac{1}{n^2}.
	\]
	
	Applying the union bound over all vertices \( v \in V \), we get that
	
	\[
	\Pr[\exists v \in V: |\tB(v)| > 16\cd \ell^{1/k} \cdot \ln n] \leq n \cdot \frac{1}{n^2} = \frac{1}{n} ~.
	\]
	
\end{proof}

\subsection{Proof of Lemma~\ref{lem:4k-3}}
\fourk*

	\begin{proof}
		Denote $\Delta=\dist(u,v)$. We will prove by induction on $0\le i\le j$, that
		\begin{equation}\label{eq:vpi}
        \dist(v,p_i(v))\leq 2i\cd \Delta~.
        \end{equation}
		
		For the base case \( i = 0 \), we have \( p_0(v) = v \), so \( \dist(v, p_0(v)) = 0 \), which satisfies \eqref{eq:vpi}.
		
		For the inductive step, assume that \eqref{eq:vpi} holds for \( 0\le i<j \), and we prove for \( i+1 \). By the minimality of $j$, we have \( p_{i}(v) \notin B(u) \), so by definition of a bunch, 
        \begin{equation}\label{eq:bunc}
        \dist(u, p_{i+1}(u)) \leq dist(u, p_{i}(v)) ~.
        \end{equation}
        Since $p_{i+1}(v)$ is the closest vertex to $v$ in $A_{i+1}$, and using the triangle inequality, we have
		\begin{eqnarray*}
		\dist(v, p_{i+1}(v)) &\leq& \dist(v, p_{i+1}(u))\\
        &\leq& \dist(v,u)+\dist(u, p_{i+1}(u))\\
        &\stackrel{\eqref{eq:bunc}}{\le}&\Delta +\dist(u, p_{i}(v))\\
		&\leq& \Delta +\dist(u,v)+\dist(v, p_{i}(v))\\
        &\stackrel{\eqref{eq:vpi}}{\le}&2\Delta+2i\cdot\Delta\\
        &=&(2i+2)\cd \Delta~,
\end{eqnarray*}
which concludes the proof of \eqref{eq:vpi}. Note that by the triangle inequality, for any $0\le i\le j$ we also have that
        \begin{equation}\label{eq:upi}
        \dist(u,p_i(v))\le \dist(u,v)+\dist(v,p_i(v))\stackrel{\eqref{eq:vpi}}{\le}\Delta+2i\cdot\Delta=(2i+1)\cdot\Delta
        \end{equation}
        
		Combining \eqref{eq:vpi} and \eqref{eq:upi}, and using that \( j \leq k-1 \), we conclude that
		\[\dist(v,p_j(v))+\dist(u,p_j(v))\leq (2j+2j+1)\cd \Delta=(4j+1)\cd \Delta\leq (4k-3)\cd \Delta~.\]
		
	\end{proof}

\subsection{Proof of Lemma~\ref{lem:2k-1}}
\twok*

\begin{proof}
		Denote $\Delta=dist(u,v)$.  We prove by induction on $0\le j\le i$ that
		\begin{eqnarray}\label{eq:indv}
		     dist(v,p_j(v))&\leq& j\cd \Delta~.\\
             \label{eq:indu}
		     dist(u,p_j(u))&\leq& j\cd \Delta~.
		\end{eqnarray}

		For the base case \( j = 0 \), we have \( p_0(v) = v \) and \( p_0(u) = u \), so both \( dist(v, p_0(v))=dist(u, p_0(u)) = 0 \), as required.
		For the inductive step, assume \eqref{eq:indv} and \eqref{eq:indu} hold for \( 0\le j<i \), and we prove for $j+1$. 
        By the minimality of $i$, we have that both $p_{j}(v) \notin B(u)$ and $p_{j}(u) \notin B(v)$. By definition of a bunch, it follows that
        \begin{equation}\label{eq:bunu}
             dist(u, p_{j+1}(u)) \leq dist(u, p_{j}(v))~, 
        \end{equation}
        and
        \begin{equation}\label{eq:bunv}
            dist(v, p_{j+1}(v)) \leq dist(v, p_{j}(u))~, 
        \end{equation}
		
        By the induction hypothesis and the triangle inequality, it follows that
        \[
		dist(v, p_{j+1}(v)) \stackrel{\eqref{eq:bunv}}{\le}dist(v, p_{j}(u))\leq \dist(v,u)+ dist(u, p_j(u))\stackrel{\eqref{eq:indu}}{\le}(j+1)\cd \Delta~,
		\] 
        and symmetrically,
\[
		dist(u, p_{j+1}(u)) \stackrel{\eqref{eq:bunu}}{\le}dist(u, p_{j}(v))\leq \dist(u,v)+ dist(v, p_j(v))\stackrel{\eqref{eq:indv}}{\le}(j+1)\cd \Delta~,
		\] 
which completes the proof of \eqref{eq:indv} and \eqref{eq:indu}. Finally, we conclude that

		\[dist(v,p_i(v))+dist(p_i(v),u)\leq dist(v,p_i(v))+dist(p_i(v),v) +\Delta\leq (2i+1)\cd \Delta~.\]

	\end{proof}
    
\end{document}